\documentclass[11pt]{article}
\usepackage[margin=1in]{geometry}
\usepackage{rpmacros}
\usepackage[T1]{fontenc}
\usepackage[usenames,dvipsnames]{color}
\usepackage{stmaryrd}
\usepackage{xfrac}
\usepackage{mathpazo}
\usepackage{todonotes}
\usepackage{amsmath}
\usepackage{setspace}

\RequirePackage[colorlinks=true]{hyperref}
\hypersetup{
  linkcolor=[rgb]{0.3,0.3,0.6},
  citecolor=[rgb]{0.2, 0.6, 0.2},
  urlcolor=[rgb]{0.6, 0.2, 0.2}
}

\usepackage{amsthm}
\usepackage{thmtools,thm-restate}

\numberwithin{equation}{section}
\declaretheoremstyle[bodyfont=\it,qed=\qedsymbol]{noproofstyle}

\declaretheorem[numberlike=equation]{observation}

\declaretheorem[name=Observation,numbered=no]{observation*}

\declaretheorem[numberlike=equation]{theorem}

\declaretheorem[name=Theorem,numbered=no]{theorem*}

\declaretheorem[numberlike=equation]{lemma}
\declaretheorem[name=Lemma,numbered=no]{lemma*}

\declaretheorem[name=Corollary,numbered=no]{corollary*}

\declaretheorem[name=Proposition,numbered=no]{proposition*}

\declaretheorem[name=Claim,numbered=no]{claim*}

\declaretheorem[name=Conjecture,numbered=no]{conjecture*}

\declaretheorem[numberlike=equation]{question}
\declaretheorem[name=Question,numbered=no]{question*}

\declaretheorem[name=Open Problem]{openproblem}

\declaretheoremstyle[bodyfont=\it,qed=$\lozenge$]{defstyle} 

\declaretheorem[numberlike=equation,style=defstyle]{definition}
\declaretheorem[unnumbered,name=Definition,style=defstyle]{definition*}

\declaretheorem[unnumbered,name=Example,style=defstyle]{example*}

\declaretheorem[unnumbered,name=Notation=defstyle]{notation*}

\declaretheorem[numberlike=equation,style=defstyle]{construction}
\declaretheorem[unnumbered,name=Construction,style=defstyle]{construction*}

\declaretheorem[unnumbered,name=Remark,style=defstyle]{remark*}

\usepackage{nth}
\usepackage{intcalc}
\usepackage{etoolbox}
\usepackage{xstring}
\hypersetup{
}

\usepackage{ifpdf}
\ifpdf
\else
\usepackage[quadpoints=false]{hypdvips}
\fi

\newcommand{\ECCCurl}[2]{http://eccc.weizmann.ac.il/report/\ifnumcomp{#1}{>}{93}{19}{20}#1/#2/}

\newcommand{\shortECCC}[2]{\texttt{\href{http://eccc.weizmann.ac.il/report/\ifnumcomp{#1}{>}{93}{19}{20}#1/#2/}{eccc:TR#1-#2}}}

\newcommand{\parseECCC}[1]{
\StrSubstitute{#1}{TR}{}[\tmpstring]%
\IfSubStr{\tmpstring}{/}{ 
\StrBefore{\tmpstring}{/}[\ecccyear]%
\StrBehind{\tmpstring}{/}[\ecccreport]%
}{
\StrBefore{\tmpstring}{-}[\ecccyear]%
\StrBehind{\tmpstring}{-}[\ecccreport]%
}%
\shortECCC{\ecccyear}{\ecccreport}}

\usetikzlibrary{decorations.pathreplacing,arrows}

\newif\ifnote
\notetrue
\ifnote
\newcommand{\BLnote}[1]{\textcolor{BrickRed}{\guillemotleft BLV: #1 \guillemotright}}
\newcommand{\MKnote}[1]{\textcolor{Purple}{\guillemotleft MK: #1 \guillemotright}}
\else
\newcommand{\BLnote}[1]{}
\newcommand{\MKnote}[1]{}
\fi

\newcommand{\ehref}[1]{\href{mailto:#1}{#1}}
\renewcommand{\E}{\mathbb{E}}

\newcommand{\C}{\mathbb{C}}
\newcommand{\spars}[1]{\left\Vert #1 \right\Vert_{0}}
\newcommand{\ip}[1]{\left\langle #1 \right\rangle}
\newcommand{\calH}{\mathcal{H}}
\newcommand{\calC}{\mathcal{C}}

\let\epsilon\varepsilon

\onehalfspace

\title{Lower Bounds for Matrix Factorization}

\author{Mrinal Kumar\thanks{\ehref{mrinalkumar08@gmail.com}. Department of Computer Science, University of Toronto, Canada. A part of this work was done during the semester on Lower Bounds in Computational Complexity at Simons Institute for the Theory of Computing, Berkeley, USA.
}
\and%
{Ben Lee Volk\thanks{\ehref{benleevolk@gmail.com}. Center for the Mathematics of Information, California Institute of Technology, USA.}
}
}

\date{}

\begin{document}

\maketitle

\begin{abstract}
We study the problem of constructing explicit families of matrices which cannot be expressed  as a product of a few sparse matrices. In addition to being a natural mathematical question on its own, this problem appears in various incarnations in computer science; the most significant being in the context of lower bounds for algebraic circuits which compute linear transformations,  matrix rigidity and  data structure lower bounds.

We first show, for every constant $d$, a deterministic construction in subexponential time of a family $\{M_n\}$ of $n \times n$ matrices which cannot be expressed as a product $M_n = A_1 \cdots A_d$ where the total sparsity of $A_1,\ldots,A_d$ is less than $n^{1+1/(2d)}$. In other words, any depth-$d$ linear circuit computing the linear transformation $M_n\cdot \vecx$ has size at least $n^{1+\Omega(1/d)}$. This improves upon the prior best lower bounds for this problem, which are barely super-linear, and were obtained by a long line of research based on the study of  super-concentrators (albeit at the cost of a blow up in the time required to construct these matrices). 

We then outline an approach for proving improved lower bounds through a certain derandomization problem, and use this approach to prove asymptotically optimal quadratic lower bounds for natural special cases, which generalize many of the common matrix decompositions.

\end{abstract}
\newpage

\section{Introduction}
\label{sec:intro}
This work concerns the following (informally stated) very natural problem:

\begin{openproblem}
\label{openproblem:lb}
Exhibit an explicit matrix $A \in \F^{n \times n}$, such that $A$ cannot be written as $A=BC$, where $B \in \F^{n \times m}$ and $C \in \F^{m \times n}$ are sparse matrices.
\end{openproblem}

Before bothering ourselves with the precise meaning of the words ``explicit'' and ``sparse'' in the above problem, we discuss the various contexts in which this problem presents itself.

\subsection{Linear circuits and matrix factorization}
\label{sec:intro:lin-circuits}

Algebraic complexity theory studies the complexity of computing polynomials using arithmetic operations: addition, subtraction, multiplication and division. An algebraic circuit over a field $\F$ is an acyclic directed graph whose vertices of in-degree 0, also called inputs, are labeled by indetermeinates $\set{x_1, \ldots, x_n}$ or field element from $\F$, and every internal node is labeled with an arithmetic operation. The circuit computes rational functions in the natural way, and the polynomials (or rational functions) computed by the circuit are those computed by its vertices of out-degree 0, called the outputs. This framework is general enough to encompass virtually all the known algorithms for algebraic computational problems. The size of the circuit is defined to be the number of edges in it. For a more detailed background on algebraic circuits, see \cite{SY10}.

Perhaps the simplest non-trivial class of of polynomials is the class of linear (or affine) functions. Accordingly, such polynomials can be computed by a very simple class of circuits called \emph{linear circuits}: these are algebraic circuits which are only allowed to use addition and multiplication by a scalar. It is often convenient to consider graphs with labels on the edges as well: every internal node is an addition gate, and for $c \in \F$, an edged labeled $c$ from a vertex $v$ to a vertex $u$ denotes that the output of $v$ is multiplied by $c$ when feeding into $u$. Thus, every node computes a linear combination of its inputs.

It is not hard to show that any arithmetic circuit for computing a set of linear functions can be converted into a linear circuit with only a constant blow-up in size (see \cite{BCS97}, Theorem 13.1; eliminating division gates requires that the field $\F$ in question is large enough. In this paper we will always makes this assumption when needed). 

Clearly, every set of $n$ linear functions on $n$ variables (represented by a matrix $A \in \F^{n \times n}$) can be computed by a linear circuit of size $O(n^2)$. Using counting arguments (over finite fields) or dimension arguments (over infinite fields), it can be shown that for a random or generic matrix this upper bound is fairly tight. Thus, a central open problem in algebraic complexity theory is to prove any super-linear lower bound for an \emph{explicit} family of matrices $\set{A_n}$ where $A_n \in \F^{n \times n}$. The standard notion of explicitness in complexity theory is that there is a deterministic algorithm that outputs the matrix $A_n$ in $\poly(n)$ time, although more or less stringent definitions can be considered as well.

Despite decades of research and partial results, such lower bounds are not known.\footnote{We remark that super-linear lower bounds for general arithmetic circuits are known, but for polynomials of high degree \cite{Strassen73, BS83}.} In order to gain insight into the general model of computation, research has focused on limited models of linear circuits, such as monotone circuits, circuits with bounded coefficients, or bounded depth circuits. We defer a more thorough discussion on previous work to \autoref{sec:intro:prev}, and proceed to describe bounded depth circuits, which are the focus of this work.

The \emph{depth} of a circuit is the length (in edges) of a longest path from an input to an output. Constant depth circuits appear to be a particularly weak model of computation. However, even this model is surprisingly powerful (see also \autoref{sec:intro:rigidity}).

The ``easiest'' non-trivial model is the model of depth-2 linear circuits. A depth 2 linear circuit computing a linear transformation $A \in \F^{n \times n}$ consists of a bottom layer of $n$ input gates, a middle layer of $m$ gates, and a top layer of $n$ output gates. We assume, without loss of generality, that the circuit is \emph{layered}, in the sense that every edge goes either from the bottom to the middle layer, or from the middle to the top layer. Indeed, every edge going directly from the bottom to the top layer can be replaced by a path of length 2; this transformation increases the size of the circuit by at most a factor of 2.

By letting $C \in \F^{m \times n}$ be the adjacency matrix of the (labeled) subgraph between the bottom and the middle layer, and $B \in \F^{n \times m}$ be the adjacency matrix as the subgraph between the bottom and the top layer, it is clear that $A=BC$. Thus, a decomposition of $A$ into the product of two sparse matrices is equivalent to saying that $A$ has a small depth-2 linear circuit. This argument can be generalized, in exactly the same way, to depth-$d$ circuits and decompositions of the form $A=A_1 \cdots A_d$, for constant $d$.

Weak super-linear lower bounds are known for constant depth linear circuits. They are based on the following observation, due to Valiant \cite{Valiant75}: for subsets $S,T\subseteq [n]$ of size $k$, let $A_{S,T}$ denote the submatrix of $A$ indexed by rows in $S$ and columns in $T$. If $A_{S,T}$ has rank $k$, the minimal vertex cut in the subcircuit restricted to input from $S$ and outputs from $T$ is of size at least $k$: indeed, a smaller cut corresponds to a factorization $A_{S,T} = PQ$ for $P \in \F^{k \times r}$ and $Q \in \F^{r \times k}$ for $r < k$, contradicting the rank assumption. Using Menger's theorem, it is now possible to deduce  that if $A$ is a matrix such that for every $S,T$ as above the matrix $A_{S,T}$ is non-singular, then the circuit computing $A$ contains, for every subcircuit which corresponds to such $S,T$, at least $k$ vertex disjoint paths from $S$ to $T$. Such graphs were named  \emph{superconcentrators} by Valiant, and their minimal size was extensively studied \cite{Valiant75, Pippenger77, Pippenger82, DDPW83, Pudlak94, AP94, RTS00}.

Superconcentrators of logarithmic depth and linear size do exist, so while this approach cannot show lower bounds for circuits of logarithmic depth, it is possible to show that for constant $d$, any depth-$d$ superconcentrator has size at least $n \cdot \lambda_d(n)$, where $\lambda_d(n)$ is a function that unfortunately grows very slowly with $n$. For example, $\lambda_2(n) = \Theta(\log^2 n / \log \log n)$, $\lambda_3(n) = \Theta(\log \log n)$, $\lambda_4(n) = \lambda_5(n) = \log^*(n)$, and so on. Such lower bounds apply for any matrix whose minors of all orders are non-zero, e.g., a Cauchy matrix given by $A_{i,j} = 1/(x_i - y_j)$ for any distinct $x_1,\ldots,x_n,y_1,\ldots,y_n$. Over finite fields it is possible to to modify the proof and obtain a similar lower bounds for matrices defining good error correcting codes \cite{GHKPV13}.

These lower bounds on the size of superconcentrators are tight: for every $d \in \N$, there exists a super-concentrator of depth $d$ and size $O(n \cdot \lambda_d (n))$. It is thus impossible to improve the lower bounds only using this technique.

\subsection{Matrix rigidity}
\label{sec:intro:rigidity}

A demonstration of the surprising power of depth-2 circuits can be seen using the notion of \emph{matrix rigidity}, a pseudorandom property of matrices which we now recall. A matrix $A \in \F^{n \times n}$ is $(r,s)$ rigid if $A$ \emph{cannot} be written as a sum $A=R+S$ where $R$ is a matrix of rank $r$, and $S$ is a matrix with at most $s$ non-zero entries. Valiant \cite{Valiant77} famously proved that if $A$ is computed by a linear circuit with bounded fan-in of depth $O(\log n)$ and size $O(n)$, then $A$ is not $(\varepsilon n, n^{1+\delta})$ rigid for every $\varepsilon,\delta>0$.\footnote{In fact, one can obtain slightly better parameters. See, for example, \cite{Valiant77} or  \cite{DGW18}.} It follows that an explicit construction $(\varepsilon n, n^{1+\delta})$ matrix, for some $\varepsilon, \delta>0$, will imply a super-linear lower bound for linear circuits of depth $O(\log n)$. Pudl{\'{a}}k \cite{Pudlak94} observed that similar rigidity parameters will imply even stronger lower bounds for constant depth circuits.
A random matrix (over infinite fields) is $(r, (n-r)^2)$-rigid, but the best explicit constructions have rigidity $(r,n^2/r \cdot \log (n/r))$ \cite{F93, SSS97}, which is insufficient for proving lower bounds.

Observe that a decomposition $A=R+S$ where $\rank(R) = \varepsilon n$ and $S$ is $n^{1+\delta}$-sparse corresponds to a depth-$2$ circuit with a very special structure and with at most $2\varepsilon n^2 + n^{1+\delta}$ edges (this circuit is not layered, but as we explained above, this does not make a significant difference). In particular, one way of interpreting Valiant's result is as a non-trivial depth reduction from depth $O(\log n)$ to depth 2, so that proving \emph{any} depth-2 $\Omega(n^2)$ lower bound for an explicit matrix, will imply a lower bound for depth $O(\log n)$.\footnote{We note that this statement makes sense only over large fields, as over fixed finite fields, it is always possible to prove an \emph{upper bound} of $O(n^2 / \log n)$ on the depth-2 complexity of any matrix  \cite{JS13}. This does not contradict the fact that rigid matrices exist over finite fields --- a decomposition to $R+S$ is a very special type of depth-$2$ circuit.} This can be seen as the linear circuit analog of similar strong depth reduction theorems for general algebraic circuits \cite{AV08, K12b, T15, GKKS16}.

However, we would like to argue that proving lower bounds for depth-2 circuits is in fact \emph{necessary} for proving rigidity lower bounds, by observing that \emph{upper bounds} on the depth-2 complexity of $A$ give upper bounds on its rigidity parameters. Indeed, suppose $A=BC$ can be computed by a depth-2 circuit of size $n^{1+\varepsilon}$. Let $m$ be as before the number of columns of $B$ (which equals the number of rows of $C$), and note that we may assume $m \le n^{1+\varepsilon}$, as zero columns of $B$ or zero rows of $C$ can be omitted. For $i \in [m]$, let $B_i$ denote the $i$-th column of $B$, and $C_i$ the $i$-th row of $C$, so that $A = \sum_{i=1}^{m} B_i C_i$. Fix a constant $\delta > 0$, and say $i \in [m]$ is \emph{dense} if either $B_i$ or $C_i$ has more than $n^{\varepsilon}/\delta$ non-zero entries; otherwise, $i$ is \emph{sparse}. Since $B$ can have at most $\delta n$ columns with sparsity of more than $n^{\varepsilon}/\delta$, and similarly for the rows of $C$, the number of dense $i$-s is at most $2 \delta n$. It follows that
\[
A = \sum_{i\text{ dense}} B_i C_i + \sum_{i\text{ sparse}} B_i C_i.
\]
The first sum is a matrix of rank at most $2 \delta n$, and the second is a matrix whose sparsity is at most $m \cdot n^{2\varepsilon}/\delta^2 = n^{1+3\varepsilon}/\delta^2$. Thus, proving rigidity lower bounds of the type required to carry out Valiant's approach necessarily means proving lower bounds of the form ``$n^{1+\varepsilon}$'' on the depth-2 complexity of $A$ (we remark that the argument above is very similar to the aforementioned result of Pudl{\'{a}}k \cite{Pudlak94}; Pudl{\'{a}}k's argument is stated in a slightly different language and in greater generality). Since proving rigidity lower bounds is a long-standing open problem, we view the problem of proving an $\Omega(n^{1+\varepsilon})$ lower bound for depth-2 circuits as an important milestone towards this.

\subsection{Data structure lower bounds}
\label{sec:intro:ds}
The problem of matrix factorization into sparse matrices also appears  in the context of proving lower bounds for data structures. A dynamic data structure with $n$ inputs and $q$ queries is a pair of algorithms whose purpose is to update and retrieve certain data under a sequence of operations, while minimizing the memory access. In the group model, it is given by a pair of algorithms. The update algorithm is represented by a matrix $U \in \F^{s \times n}$. Given $x \in \F^{n}$, thought of as assignment of weights to the $n$ inputs, $Ux$ computes a linear combination of those weights and stores them in memory. The query algorithm is given by a matrix $Q \in \F^{q \times s}$. Given a query, it computes a linear function of the $s$ memory cells, and returns the answer. Hence, an ``update'' operation followed by a ``retrieve'' operation computes the linear transformation given by $A=QU$.

The worst case update time of the database is the maximal number of non-zero elements in a column of $U$, and the worst case query time is the maximal number of non-zero elements in a row of $Q$. The value $s$ denotes the space required by the data structure. It now directly follows that a matrix $A \in \F^{q \times n}$ which cannot be factored as $A =QU$ for a row-sparse $Q$ and column-sparse $U$ gives a data structure problem with a lower bound on its worst case query or update time. It is also possible to define an analogous average case notion. Lower bounds for this model were proved by \cite{Fredman82, FredmanSaks89, PD06, Patrascu07, Larsen12, Larsen14, LWY18}, but none of these results beats the lower bounds for depth-2 circuits obtained using superconcentrators. 

A related model is that of a static data structures, which is again given by a factorization $A=QP$, where now we are interested in trade-offs between the space $s$ of the data structure and its worst case query time, while not being charged for the total sparsity of $P$. A recent work of Dvir, Golovnev and Weinstein \cite{DGW18} showed that proving lower bounds for this model is related to the problem of matrix rigidity from \autoref{sec:intro:rigidity}.

Despite the overall similarity, there are several key technical differences between the linear circuit complexity and the data structure problems. The first and obvious issue is that worst-case lower bounds on the update or query time do not necessarily imply that $Q$ or $U$ are dense matrices: the total sparsity of $Q$ and $U$ is related to the average-case update and query time. The second, more severe issue, is that in many applications the number of queries $q$ is polynomially larger than $n$, while the lower bounds on running time are still measured as functions of the number of inputs $n$. This makes sense in the data structure settings, but from a circuit complexity point of view, a set of say $n^3$ linear functions trivially requires a circuit of size $n^3$, and thus a lower bound of say $n \polylog(n)$ is meaningless in that setting.

This issue also comes up when studying the so-called \emph{succinct space} setting, where we require $s=n(1+o(1))$. The lower bounds we are aware of for this setting are worst case lower bounds, and require the number of outputs $q$ to be at least $Cn$ for some $C>1$ \cite{GM07,DGW18}, so that in the corresponding circuit the number of vertices in the middle layer is required to be much smaller than the number of outputs, which may be considered quite unnatural. In particular, we are unaware of any improved lower bounds on the sparsity of matrix factorization for $A \in \F^{n \times n}$ when $s=n(1+o(1))$ or even $s=n$ which come from the data structure lower bounds literature.

\subsection{Machine learning}
\label{sec:intro:ml}
We briefly remark that the problem of factorizing a matrix into a product of two or more sparse matrices is also ubiquitous in machine learning and related areas. Naturally, research in those areas did not focus on lower bounds but rather on algorithms for finding such a representation, assuming it exists, sometimes heuristically, and it is usually enough to approximate the target matrix $A$. In particular, algorithms have been proposed for the very related problems of non-negative matrix factorization \cite{LS00}\footnote{It is interesting to observe that for the problem of factorizing matrices into non-negative matrices it is quite easy to prove almost-optimal lower bounds even for unbounded depth linear circuits, as mentioned in \autoref{sec:intro:prev}} or sparse dictionary learning \cite{MBPS09}, and there are also connections to the analysis of deep neural networks \cite{NP13}.

\subsection{Previous work}
\label{sec:intro:prev}

As mentioned in \autoref{sec:intro:lin-circuits}, there are no non-trivial known lower bounds for general linear circuits, and for bounded depth circuits, the best lower bounds follow from the lower bounds on bounded depth super-concentrators, which are barely super-linear. 

Shoup and Smolensky \cite{SS96} give a lower bound of $\Omega(dn^{1+1/d})$ for depth-$d$ circuits computing a certain linear transformation given by a matrix $A \in \R^{n \times n}$. Unfortunately, the matrices for which their lower bound holds are not explicit from the complexity theoretic point of view, despite having a very  succinct mathematical description (for example, one can take $A_{i,j} = \sqrt{p_{i,j}}$ for $n^2$ distinct prime numbers $p_{i,j}$). For the same matrix, they in fact prove super-linear lower bounds for circuits of depth up to $\polylog(n)$.

Quite informally, the intuition behind their lower bounds is that all small bounded depth linear circuits can be described as lying in the image of a low-degree polynomial map in a small number of variables, and thus, if the elements of $A$ are sufficiently ``algebraically rich'', for a certain specific measure, $A$ cannot be computed by such a circuit. This same philosophy lies behind Raz's elusive function approach for proving lower bounds for algebraic circuits \cite{Raz10a}. In particular, among other results, Raz uses an argument which can be seen as a modification of the technique of Shoup and Smolensky (as worked out in \cite{SY10}) to prove lower bounds for bounded depth algebraic circuits computing bounded degree polynomials.

One class of linear circuits which has attracted significant attention is the class of circuits with bounded coefficients. Here, the circuit is only allowed to multiply by scalars with absolute value of at most some constant. For definiteness, we may assume this constant is 1 (this does not affect the complexity by more than a constant factor). The earliest result for this model is Morgenstern's ingenious proof \cite{Morgenstern73} of an $\Omega(n \log n)$ lower bound on bounded coefficient circuits computing the discrete Fourier transform matrix (this lower bound is matched by the upper bound given by the Cooley-Tukey FFT algorithm, which is a bounded coefficient linear circuit). For depth-$d$ circuits, Pudl{\'{a}}k \cite{Pudlak00} has proved lower bounds of the form $\Omega(d n^{1+1/d})$ for the same matrix.

Another natural subclass which was considered in earlier works is the class of monotone linear circuits. These are circuits which are defined over $\R$, and can only use non-negative scalars. Chazelle \cite{Chazelle2001} observed that it is possible to prove lower bounds in this model, even against unbounded-depth circuits, for any boolean matrix with no large monochromatic rectangle. Instantiated with the recent explicit constructions of bipartite Ramsey graphs \cite{CZ16, BDT17, Cohen17, Li18}, this gives an almost optimal $n^{2-o(1)}$ lower bound against such circuits. The main observation in the proof is that if $A$ does not have monochromatic $t \times t$ rectangle, then since the model is monotone and no cancellations are allowed, every internal node which computes a linear function supported on at least $t$ variables cannot be connected to more than $t$ output gates.

For a more detailed survey on these results and some other related results, see the survey by Lokam \cite{Lokam09}.

\subsection{Our results}
\label{sec:intro:results}
In this paper, we prove several results regarding bounded depth linear circuits which we now discuss.
\paragraph{Lower bounds for depth-$d$ linear circuits. }We start by considering general depth-$d$ circuits. We construct, in subexponential time, matrices which require depth-$d$ circuits of size $n^{1+\Omega(1/d)}$.

\begin{theorem}
\label{thm:intro-depth-d}
Let $\F$ be a field. There exists a family of matrices $\set{A_n}_{n \in \N}$, which can be constructed in time $\exp(n^{1-\Omega(1/d)})$, such that every depth-$d$ linear circuit computing $A_n$, even over the algebraic closure of $\F$, has size at least $n^{1+\Omega(1/d)}$.

If $\F=\Q$, the entries of $A$ are integers of bit complexity $\exp(n^{1-\Omega(1/d)})$. If $\F=\F_q$ is a finite field, the entries of $A$ are elements of an extension $\E$ of $\F$ of degree $\exp(n^{1-\Omega(1/d)})$.
\end{theorem}

This theorem is proved in \autoref{sec:shoup-smol based lb}. We remark again that the best lower bounds against general depth-$d$ linear circuits for matrices that can be constructed in polynomial time are barely super-linear and much weaker than $n^{1+\varepsilon}$. In the recent work of Dvir, Golovnev and Weinstein \cite{DGW18} it was pointed out that currently there are not even known constructions of rigid matrices (with parameters that would imply lower bounds) in classes such as $\mathbf{E}^{\mathbf{NP}}$. By arguing directly about circuit size, and not about rigidity, \autoref{thm:intro-depth-d} gives constructions of matrices in a much smaller complexity class, which have the same bounded-depth complexity lower bounds as would follow from optimal constructions of rigid matrices using the results of Pudl{\'{a}}k \cite{Pudlak94}.

While the statement in \autoref{thm:intro-depth-d} holds for any $d \ge 2$, for $d=2$ there is a much simpler construction of a hard family of matrices in quasi-polynomial time. 
\begin{theorem}\label{thm:intro-depth-2-quasipoly}
Let $\F$ be any field and $c$ be any positive constant. Then, there is a family $\{A_n\}_{n \in \N}$ of $n \times n$ matrices which can be constructed in time $\exp(O(\log^{2c + 1} n))$ such that any depth-$2$ linear circuit computing $A_n$ even over the algebraic closure of $\F$ has size at least $\Omega(n\log^c n)$. 
\end{theorem}

For every constant $c \geq 2$, this theorem already improves upon the current best lower bound of $\Omega(n\log^2 n/\log\log n)$ known for this problem (see~\cite{RTS00}). This construction is based on an exponential time construction of a small hard matrix, and then amplifying its hardness using a direct sum construction (note, however, that over infinite fields even the fact that a hard matrix can be constructed in exponential time, while not very hard to prove, is not \emph{completely} obvious).
For completeness, we describe this simple construction in~\autoref{subsec:depth-2-direct-sum}.

\paragraph*{Lower bounds for restricted depth-$2$ linear circuits. }

Given the importance of the model of depth-2 linear circuits, as explained above, and its resistance to strong lower bounds, we then move on to consider several natural subclasses of depth-2 circuits. These classes in particular correspond to almost all common matrix decompositions. We are able to prove asymptotically optimal $\Omega(n^2)$ lower bounds for these restricted models. As mentioned above, such lower bounds for general depth-2 circuits will imply super-linear lower bounds for logarithmic depth linear circuits, thus resolving a major open problem.

\paragraph{Symmetric circuits. }
A symmetric depth-2 circuit (over $\R$) is a circuit of the form $B^T B$ for some $B \in \R^{m \times n}$ (considered as a graph, the subgraph between the middle and the top layer is the ``mirror image'' of the subgraph between the bottom and middle layer). Over $\C$, one should take the conjugate transpose $B^*$ instead of $B^T$. 

Symmetric circuits are a natural computational model for computing positive semi-definite (PSD) matrix. Clearly, every symmetric circuit computes a PSD matrix, and every PSD matrix has a (non-unique) symmetric circuit. In particular, a Cholesky decomposition of PSD matrices corresponds to a computation by a symmetric circuit (of a very special form).

We prove asymptotically optimal lower bounds for this model.

\begin{theorem}
\label{thm:intro:PSD}
There exists an explicit family of real $n \times n$ PSD matrices $\set{A_n}_{n \in \N}$ such that every symmetric circuit computing $A_n$ (over $\R$ or $\C$) has size $\Omega(n^2)$.
\end{theorem}

We do not know whether every depth-2 linear circuit for a PSD matrix can be converted to a symmetric circuit with a small blow-up in size. One way to phrase this question is given below.

\begin{question}
\label{ques:intro:PSD}
Is there a constant $c<2$, such that every PSD matrix $A \in \R^{n \times n}$ which can be computed by a linear circuit of size $s$, can be computed by a symmetric circuit of size $O(s^c)$? 
\end{question}

A positive answer for \autoref{ques:intro:PSD} will imply, using \autoref{thm:intro:PSD}, an $\Omega(n^{1+\varepsilon})$ lower bound for depth-2 linear circuits.

\paragraph{Invertible circuits. }
Invertible circuits are circuits of the form $BC$, where either $B$ or $C$ are invertible (but not necessarily both). We stress that invertible circuits can (and do) compute non-invertible matrices. In particular, if $B \in \F^{n \times m}$ and $C \in \F^{m \times n}$, here we require $m=n$.

Invertible circuits generalize many of the common matrix decompositions, such as QR decomposition, eigendecomposition, singular value decomposition\footnote{A diagonal matrix can be multiplied with the matrix to its left or to its right, without increasing the sparsity, to obtain an invertible depth-$2$ circuit.} and LUP decomposition (in the case where the matrix $L$ is required to be unit lower triangular).\footnote{The sparsity of $UP$ equals the sparsity of $U$, as $P$ simply permutes the columns of $U$, so every $LUP$ decomposition corresponds to the invertible depth-$2$ circuit given by $L(UP)$.}

We prove optimal lower bounds for invertible circuits.

\begin{theorem}
\label{thm:intro:invertible}
Let $\F$ be a large enough field. There exists an explicit family of $n \times n$ matrices $\set{A_n}_{n \in \N}$ over $\F$ such that every invertible circuit computing $A_n$ has size $\Omega(n^2)$.
\end{theorem}

If $A$ is an invertible matrix, then clearly every depth-$2$ circuit with $m=n$ must be an invertible circuit. However, our technique for proving \autoref{thm:intro:invertible} crucially requires the hard matrix $A$ to be non-invertible.

\subsection{Proof Overview}
\label{sec:techniques}

Our proofs rely on a few different ideas coming from algebraic complexity theory, coding theory, arithmetic combinatorics and the theory of derandomization. We now discuss some of the  key aspects.

\paragraph*{Shoup-Smolensky dimension.}
For the proof of \autoref{thm:intro-depth-d}, we rely on the notion of \emph{Shoup-Smolensky} dimension as a measure of complexity of matrices. Shoup-Smolensky dimensions are a family of measures, parametrized by $t \in \N$, of ``algebraic richness'' of the entries of a matrix (see \autoref{def:SS-dim} for details), which is supposed to capture the intuition that matrices with small circuits should depend on a few ``parameters'' and thus should not posses much richness. 

Shoup and Smolensky~\cite{SS96} showed that for an appropriate choice of parameters, this measure is non-trivially small for linear transformations with small linear circuits of depth at most $\poly(\log n)$. Informally, as the order $t$ gets larger, this measure becomes useful against stronger models of computation; however, it also becomes harder to construct matrices which have a large complexity with respect to this measure (and hence cannot be computed by a small linear circuit). Shoup and Smolensky do this by constructing hard matrices which do not have small bit complexity (and hence this construction is not complexity theoretically explicit) but do have short and succinct mathematical description.

For our proof, we first observe that for bounded depth circuits it suffices to use much smaller order $t$ than what Shoup and Smolensky used. This observation was also made by Raz \cite{Raz10a} in a similar context, but in a different language. 

We then use this observation to ``derandomize'', in a certain sense, an exponential time construction of a hard matrix, by giving deterministic constructions of matrices with large Shoup-Smolensky dimension. 

A key ingredient of our proof is a connection between the notion of Sidon Sets in arithmetic combinatorics and Shoup-Smolensky dimension (see~\autoref{sec:sidon-ss dim} for details).  Our construction is in two steps.  In the first step we construct matrices with entries in $\F[y]$ which have a large Shoup-Smolensky dimension over $\F$, and degree of every entry is not too large. In the next step, we go from these univariate matrices to a matrix with entries in an appropriate low degree extension of $\F$ while still maintaining the Shoup-Smolensky dimension over $\F$. Our construction  of hard matrices over the field of complex numbers is based on similar ideas but differs in some minor details.

\paragraph*{Lower bounds via Polynomial Identity Testing. } Our proofs for \autoref{thm:intro:PSD} and \autoref{thm:intro:invertible} are based on a derandomization argument. Connections between derandomization and lower bounds are prevalent in algebraic and Boolean complexity, but in our current setting they have not been  widely studied before.

We say that a set $\calH$ of $n \times n$ matrices is a \emph{hitting set}
for a class $\calC$ of matrices if for every non-zero $A \in \calC$ there is $H \in \calH$ such that $\ip{A,H} := \sum_{i,j} A_{i,j}H_{i,j} \neq 0$.

Every class $\calC$ has a hitting set of size $n^2$, namely the indicator matrices of each of the entries. A hitting set is non-trivial if its size is at most $n^2 - 1$. Observe that a non-trivial hitting set for $\calC$ gives an efficient algorithm for finding a matrix $M \not\in \calC$, by finding a non-zero $A$ such that $\ip{A,H} = 0$ for every $H \in \calH$. Such an $A$ exists and can be found in polynomial time because the set $\calH$ imposes at most $n^2 - 1$ homogeneous linear constraints on the $n^2$ entries of $A$. This argument is a special case of a more general theorem showing how efficient algorithms for black box polynomial identity testing give lower bounds for algebraic circuits \cite{A05a, HS80}.

In practice, it is often convenient (although by no means necessary) to consider hitting sets that contain only rank 1 matrices $\vecx \vecy^T$, since $\ip{A,\vecx \vecy^T} = \vecx^T A \vecy$, and thus we find ourselves in the more familiar territory of polynomial identity testing, trying to construct a hitting set for the class of polynomials of the form $\vecx^T A \vecy$ for $A \in \calC$. This approach was also taken by Forbes and Shpilka \cite{FS12}, who considered this exact problem where $\calC$ is the class of low-rank matrices, and remarked that hitting sets for the class of low-rank matrices plus sparse matrices will give an explicit construction of a rigid matrix.

We carry out this idea for two different classes in the proofs of \autoref{thm:intro:PSD} and \autoref{thm:intro:invertible}. However, the following problem remains open.

\begin{openproblem}
\label{openproblem:hit-sparse}
For some $0<\epsilon \le 1$, construct an explicit hitting set of size at most $n^2 - 1$ for the class of $n \times n$ matrices $A$ which can be written as $A=BC$ where $B,C$ have at most $n^{1+\epsilon}$ non-zero entries.
\end{openproblem}

A solution to \autoref{openproblem:hit-sparse} will imply lower bounds of the form $n^{1+\varepsilon}$ for an explicit matrix. If $\epsilon=1$, this will imply lower bounds for logarithmic depth linear circuits.

A useful ingredient in our constructions is the use of maximum distance separable (MDS) codes (for example, Reed-Solomon codes), as their dual subspace is a small dimensional subspace which does not contain sparse non-zero vectors. Over the reals, it is also easy to give such construction based on the well known Descartes' rule of signs which says that a sparse univariate real polynomial cannot have too many real roots. We refer the reader to~\autoref{sec:hitting set const} for details.

\section{Lower bounds for constant depth linear circuits}\label{sec:shoup-smol based lb}
In this section, we prove~\autoref{thm:intro-depth-d}. 
We start by describing  the notion of Shoup-Smolensky dimension, but first we set up some notation.

\subsection{Notation}
 We work with matrices whose entries lie in an appropriate extension of a base finite field $\F_p$. We follow the natural convention that the elements of this extension will be represented as univariate polynomials of appropriate degree over the base field, and the arithmetic is done modulo an explicitly given irreducible polynomial. 

We use boldface letters ($\vecx,\vecy$) to denote vectors. The length of the vectors is understood from the context.

For a matrix $M$, $\spars{M}$ denotes the number of non-zero entries in $M$. 

\subsection{Shoup-Smolensky Dimension}
\label{sec:ssdim}
A useful concept will be the notion of Shoup-Smolensky dimension of subsets of elements of an extension $\E$ of a field $\F$. 
\begin{definition}[Shoup-Smolensky dimension]
\label{def:SS-dim}
Let $\F$ be a field, and $\E$ be an extension field of $\F$. 
Let $M \in \E^{n\times n}$ be a matrix. For $t \in \N$, denote by $\Pi_t(M_L)$ the set of $t$-wise products of distinct entries of $M$ that is,
\[
\Pi_t(M) = \set{\prod_{(a,b) \in T} M_{a,b} : T \in \binom{[n] \times [n]}{t}}.
\]

The \emph{Shoup-Smolensky dimension} of $M$ of order $t$, denoted by $\Gamma_{t, \F}(M)$ is defined to be the dimension, over $\F$, of the vector space spanned by $\Pi_t(M)$.

We also denote by $\Sigma_{t}(M)$ the number of distinct elements of $\E$ that can be obtained by \emph{summing} distinct elements of $\Pi_t(M)$.
\end{definition}

\subsection{Upper bounding the Shoup-Smolensky dimension for Sparse Products}

\label{sec:SS-dim}
The following lemma shows that any matrix computable by a depth-$d$ linear circuit of size at most $s$ has a somewhat small Shoup-Smolensky dimension. 

\begin{lemma}\label{lem:ss-easy-ub-large-depth}
Let $\F$ be a field, $\E$ an extension of $\F$ and $A \in \E^{n \times n}$ be a matrix such that $A = \prod_{i = 1}^d P_i$ for $P_i \in \E^{n_i \times m_i}$, where $\sum_{i = 1}^d\spars{P_i} \leq s$. Then, for every $t\le n^2/4$ such that $s \ge dt$ it holds that
\[
\Gamma_{t, \F}(A) \le \inparen{e^d (2s/dt)^d}^t.
\]
\end{lemma}
\begin{proof}
Since 
$$A_{i,j} = \left(\prod_{\ell = 1}^d P_{\ell}\right)_{i,j} = \sum_{k_1, \ldots, k_{d-1}} (P_1)_{i, k_1}\cdot \left(\prod_{\ell = 2}^{d-1} (P_{\ell})_{k_{\ell-1}, k_{\ell}} \right) \cdot (P_{d})_{k_{d-1}, j}\, , $$
every element in $\Pi_{t} (A)$ is a sum of monomials of degree $dt$ in the entries of $P_1, P_2, \ldots, P_d$, that is,
\[
\Gamma_{t, \F}\left(\prod_{i = 1}^d P_i\right) \leq \binom{s + dt}{dt},
\]
with the right hand side being the number of monomials of degree $dt$ in $s$ variables.
Using the inequality $\binom{n}{k} \le (en/k)^k$,
\[
\Gamma_{t, \F}(A) \leq (e(1 + s/dt))^{dt} \le \inparen{e^d (2s/dt)^d}^t. \qedhere
\]
\end{proof}

Over $\Q$, we do not wish to use field extensions (which would give rise to elements with infinite bit complexity). Thus, we use a similar argument that replaces the measure $\Gamma_{t,\F}$ with $\Sigma_{t}$ (recall \autoref{def:SS-dim}) for a small tolerable penalty.
\begin{lemma}
\label{lem:ss-up-sigma}
Let $d$ be a positive integer. Let $A \in \Q^{n \times n}$ be a matrix such that $A = \prod_{i = 1}^d P_i$ for $P_i \in \Q^{n_i \times m_i}$, where $\sum_{i = 1}^d\spars{P_i} \leq s $. Assume that for each $i$,  $n_i \leq n^2$ and $m_i \leq n^2$. Then, for every $t\le n^2/4$ such that $s \ge dt$ it holds that
\[
\Sigma_{t}(A) \le 2^{2n^3\cdot \inparen{e^d (2s/dt)^d}^t}.
\]

\end{lemma}

\begin{proof}
We follow the same steps as in the proof of~\autoref{lem:ss-easy-ub-large-depth}, replacing the measure $\Gamma_{t,\F}(A)$ by $\Sigma_t(A)$. As before, 
\[
A_{i,j} = \left(\prod_{\ell = 1}^d P_{\ell}\right)_{i,j} = \sum_{k_1, \ldots, k_{d-1}} (P_1)_{i, k_1}\cdot \left(\prod_{\ell = 2}^{d-1} (P_{\ell})_{k_{\ell-1}, k_{\ell}} \right) \cdot (P_{d})_{k_{d-1}, j}\, .
\]
Every element in $\Pi_{t} (A)$ can be written as 
\begin{equation}
\label{eq:monomials-in-product}
\sum_{\alpha \in \mathcal{M}} c_\alpha \cdot \alpha
\end{equation}
where $\mathcal{M}$ is the set of monomials of degree $dt$ in the entries of $P_1, P_2, \ldots, P_d$, and each $c_\alpha$ is a non-negative integer of of absolute value at most $s^{dt} \le 2^{n^3}$ (since $s \leq n^2d$ and $d$ is $O(1)$). It now follows that each element in $\Sigma_t(A)$ has the same form as in \eqref{eq:monomials-in-product}, with $c_\alpha \le |\Pi_t(A)| \cdot 2^{n^3} \le 2^{2n^3}$ . We conclude that
\[
\Sigma_t(A) \le (2^{2n^3})^{\binom{s+dt}{dt}},
\]
which implies the statement of the lemma using the same bounds on binomial coefficients as in \autoref{lem:ss-easy-ub-large-depth}.
\end{proof}

We now move on to describe constructions of matrices which have large Shoup-Smolensky dimension, and then deduce lower bounds for them.

\subsection{Sidon sets and hard univariate matrices}\label{sec:sidon-ss dim}
In this section, we describe a construction of a matrix $G \in \F[y]^{n \times n}$ which has a large value of $\Gamma_{t, \F}$. Let us denote $G_{i,j} = y^{e_{i,j}}$ for some non-negative integer $e_{i,j}$. For $G$ to have a large Shoup-Smolensky dimension of order $t$, the set $S = \set{e_{1,1}, e_{1,2}, \ldots, e_{n,n}} \subseteq \N$ should have the property that $tS := \set{a_1 + a_2 + \ldots + a_t : a_i \in S \text{ distinct}}$ has size comparable to $\binom{|S|}{t}$. A set $S$ such that every subset of size $t$ of $S$ has a distinct sum is called a \emph{$t$-wise Sidon set}. These are very well studied objects in arithmetic combinatorics, and explicit constructions are known for them in $\poly(n)$ time (e.g., Lemma 60 in~\cite{Bshouty}). However, another important parameter in the construction is the degree of $y$, and such a set will inevitably contain integers of size roughly $n^{\Omega(t)}$. Thus, the construction of $G$ would take time which is not polynomially bounded in $n$. Below we give an elementary construction of such a set in time $n^{O(t)}$ (cf.\ \cite{AGKS15}).

\begin{lemma}
\label{lem:easy-sidon}
Let $t$ be a positive integer. There is a set $S = \set{e_{i,j} : i,j \in [n]} \subseteq \N$ of size $n^2$ such that:
\begin{enumerate}
\item $tS := \set{a_1 + a_2 + \ldots + a_t : a_i \in S \text{ distinct}}$ has size $\binom{n^2}{t}$.
\item $\max_{i,j \in [n]} \{ e_{i,j} \} \le n^{O(t)}$.
\item $S$ can be constructed in time $n^{O(t)}$.
\end{enumerate}
\end{lemma}

\begin{proof}
Let $S' = \set{1,2,2^2,\ldots,2^{n^2-1}}$. Clearly, every subset of $S'$ has a distinct sum. For a prime $p$ we denote $S_p = S' \bmod p = \set{a \bmod p : a \in S'}$, and we claim that there exists a prime $p \le n^{O(t)}$ such that $|tS_p| = \binom{n^2}{t}$. Since this condition can be checked in time $n^{O(t)}$, this would immediately imply the statement of the lemma, by checking this condition for every $p \le n^{O(t)}$ and letting $S=S_p$ for a $p$ which satisfies this condition.

For every subset $T \subseteq S'$ of size $t$, let $\sigma_T$ denote the sum of its elements, and observe that $\sigma_T \le 2^{n^2}$. Clearly, $\sigma_T \bmod p = \sigma_{T'} \bmod p$ if and only if $p \mid \sigma_T - \sigma_{T'}$, so it is enough to show that there exists $p \le n^{O(t)}$ which does not divide
\[
N := \prod_{\substack{T \neq T' \subseteq S' \\ |T|=|T'|=t}} (\sigma_T - \sigma_{T'}),
\]
and therefore does not divide any of the terms on the right hand size.
It further holds that $0 \neq N \le {(2^{n^2})}^{n^{O(t)}} = 2^{n^{O(t)}}$, so the existence of $p$ now follows from the fact that $N$ can have at most $\log N = n^{O(t)}$ distinct prime divisors, and from the prime number theorem.
\end{proof}

Given the above construction of $t$-wise Sidon sets, we now describe the construction of matrices with univariate polynomial entries which has large Shoup-Smolensky dimension.

\begin{construction}
\label{con:univariate hard matrix}
Let $S = \set{e_{i, j} : i, j \in [n]}$ be a $t$-wise Sidon set of positive integers, as in \autoref{lem:easy-sidon}. Then, the matrix $G_{t,n} \in \F[y]^{n \times n}$ is defined as follows as 
$(G_t)_{i, j} = y^{e_{i, j}}$.
\end{construction}

The useful properties of \autoref{con:univariate hard matrix} are given by the following lemma.

\begin{lemma}
\label{lem:univariate hard properties}
Let $t \le n$ be a parameter, $S \subseteq N$ be a $t$-wise Sidon set of size $n^2$ and let $G_{t,n}$ be the matrix defined in~\autoref{con:univariate hard matrix}. Then, the following are true. 
\begin{enumerate}
\item Every entry of $G_{t,n}$ is a monomial of degree at most $n^{O(t)}$. 
\item $\Gamma_{t,\F}((G_{t,n})) \geq \binom{n^2}{t} \ge  \left(\frac{n^2}{t}\right)^t$.
\end{enumerate}
\end{lemma}
\begin{proof}
The first item follows from the definition of $G_{t,n}$ and the properties of the set $S$ in~\autoref{lem:easy-sidon}. The second item also follows from the properties of $S$ and the definition of Shoup-Smolensky dimension, since every $t$-wise product of elements of $G_{t,n}$ gives a distinct monomial in $y$, and thus they are all linearly independent over the base field $\F$.
\end{proof}

\subsection{Hard matrices over finite fields}
From the univariate matrix in \autoref{con:univariate hard matrix}, we  now construct, for every $p$ and parameter $t$, a matrix $M$ over an extension of $\F_p$ which has large Shoup-Smolensky dimension over $\overline{\F}_p$  with the same parameters as $G_{t,n}$.
 
\begin{lemma}
\label{lem:hard-over-finite}
Let $p$ be a prime, and $t $ be any positive integer. There is a matrix $M_{t,n} \in \E^{n\times n}$ over an extension $\E$ of $\F_p$ of degree $\exp\inparen{{O(t\log n)}}$, which can be deterministically constructed in time $n^{O(t)}$, and satisfies
\[
\Gamma_{t, \F_p}(M_{t,n}) \geq \left(\frac{n^2}{t}\right)^t\, 
\] 
\end{lemma}
\begin{proof}
Let $G_{t,n}$ be as in \autoref{con:univariate hard matrix}, and let $\Delta$ be the maximum degree of any entry of $G_{t,n}$. Set $D = 10\cdot t\cdot \Delta = \exp\left(O(t\log n)\right)$. We use Shoup's algorithm (see Theorem 3.2 in~\cite{Shoup90}) to construct an irreducible polynomial $g(z)$ of degree $D+1$ over $\F_p$ in deterministic  $\poly(D, |\F_p|)$ time. Let $\alpha$ be a root of $g(z)$ in an extension $\E$ of $\F_p$, where $\E \equiv \F_p[z]/\langle g(z) \rangle$.\footnote{We identify the elements of $\E$ with coefficient vectors of polynomials of degree at most $D$ in $\F_p[z]$, and in this representation $\alpha$ is identified with the polynomial $z$.} Then, it follows that $1, \alpha, \alpha^2, \ldots, \alpha^{D}$ are linearly independent over $\F$.

The matrix $M_{t,n}$ is obtained from $G_t$ by just replacing every occurrence of the variable $y$ by $\alpha$. We now need to argue that $M_{t,n}$ continues to satisfy $\Gamma_{t, \F_p}(M_{t,n}) \geq \left(\frac{n^2}{t}\right)^t$. By the choice of $\alpha$, it immediately follows that $\Gamma_{t, \F_p}(M_{t,n}) = \Gamma_{t, \F_p}(G_{t,n})$, since every monomial in the set $\Pi_t(M_{t,n})$ is mapped to a distinct power of $\alpha$ in $\{0, 1, \ldots, D\}$, which are all linearly independent over $\F_p$.

The upper bound on the running time needed to construction $M_{t,n}$ now follows from the upper bound on the degree of the extension $\E$, and from \autoref{lem:easy-sidon}.
\end{proof}
The following theorem now directly follows.
\begin{theorem}\label{thm:depth-d-finite}
Let $p$ be any prime and $d\geq 2$ be a positive integer. Then, there exists a family of matrices $\{A_n\}_{n \in \N}$ which can be constructed in time $n^{O(n^{1-1/2d})}$ such that every depth-$d$ linear circuit $\overline{\F}_p$ computing $A_n$ has size at least $\Omega(n^{1 + 1/2d})$. Moreover, the entries of $A_n$ lie in an extension of $\F_p$ of degree at most $\exp(O(n^{1-1/2d}\log n))$. 
\end{theorem}
\begin{proof}
We invoke~\autoref{lem:hard-over-finite} with parameter $t$ set to  $n^{1-1/2d}$ to get matrices $\{A_n\}$ in time $n^{O(t)}$ with the following lower bound on their Shoup-Smolensky dimension. 
\[
\Gamma_{t, \F_p}(M_{n}) \geq \left(\frac{n^2}{t}\right)^t\, .
\] 
If there is a  depth $d$ linear circuit of size $s$ computing the linear transformation $A_n\cdot \vecx$, the following inequality must hold (from~\autoref{lem:ss-easy-ub-large-depth}),
\begin{equation}\label{eqn:depth-d-lb}
\inparen{e^d (2s/dt)^d}^t \geq \left(\frac{n^2}{t}\right)^t \, .
\end{equation}
If $s \leq   n^{1+1/2d}/2$, we have, 
\[
\inparen{e^d (2s/dt)^d}^t \leq (O(e/d))^{dt} \cdot n^t\, . 
\]
We also have, 
\[
\left(\frac{n^2}{t}\right)^t \geq \left(n^{1 + 1/2d} \right)^t \, .
\]
For any constant $d$, these estimates contradict~\autoref{eqn:depth-d-lb}, thereby implying a lower bound of $\Omega(n^{1+1/2d})$ on s. 
\end{proof}

\subsection{Hard matrices over  $\C$}
We now prove an analog for \autoref{lem:hard-over-finite}. We construct a matrix  whose entries are positive integers that can be represented by at most $\exp(O(t\log n))$ bits, and give a lower bound for its $\Sigma_t$-measure (rather than $\Gamma_{t,\F}$ as before).

\begin{lemma}
\label{lem:hard-over-C}
Let $t$ be any positive integer. There is a matrix $M_{t,n} \in \Q^{n\times n}$, which can be deterministically constructed in time $n^{O(t)}$, such that every entry of $M_{t,n}$ is an integer of bit complexity at most $\exp(O(t \log n))$, and it holds that
\[
\Sigma_{t}(M_{t,n}) \geq 2^{\left(\frac{n^2}{t}\right)^t}. 
\] 
\end{lemma}

\begin{proof}
Let $G_{t,n} \in \F[y]^{n \times n}$ be as in \autoref{con:univariate hard matrix}. Define $M_{t,n} \in \Q^{n \times n}$ as
\[
(M_{t,n})_{a,b} = (G_{t,n})_{a,b} (2),
\]
that is, $(M_{t,n}){a,b}$ is simply the polynomial $(G_{t,n})_{a,b}(y)$ evaluated at $y=2$.

As in the proof of \autoref{lem:univariate hard properties}, each element in $\Pi_t(M_{t,n})$ is now a distinct power of 2, which implies that $\Sigma_t(M_{t,n}) = 2^{\binom{n^2}{t}}$.

The statement on the running time follows directly from \autoref{lem:univariate hard properties}.
\end{proof}

The analog of \autoref{thm:depth-d-finite} for $\C$ is given below.

\begin{theorem}~\label{thm:depth-d-complex}
There exists a family of matrices $\{A_n\}_{n \in \N}$ over $\Q$ which can be constructed in time $n^{O(n^{1-1/2d})}$ such that every depth-$d$ linear circuit $\C$ computing $A_n$ has size at least $\Omega(n^{1 + 1/2d})$. Moreover, the entries of $A_n$ are positive integers of bit complexity at most $\exp(O(n^{1-1/2d}\log n))$. 
\end{theorem}
\begin{proof}
Let $s = n^{1+1/2d}/2$ and $t = n^{1 -1/2d}$ and let $A_n = M_{t,n}$, where $M_{t,n}$ is as in \autoref{lem:hard-over-C}. A depth-$d$ circuit for $M_n$ implies a factorization $M_n = \prod_{i =1}^d P_i$, with $P_i \in \C^{n_i \times m_i}$, such that $\sum_{i = 1}^d\spars{P_i} \le s$. Observe that since zero columns of $P$ or zero rows of $Q$ can be omitted without affecting the product, we may assume $n_i, m_i \le n^2$, as otherwise the lower bound trivially holds.
By \autoref{lem:ss-up-sigma} and \autoref{lem:hard-over-C}, this implies that
\[
(n^2/t)^t \leq \log \Sigma_{t} (A_n) \le 2n^3 \cdot \inparen{e^d (2s/t)^d}^t.
\]
If $s \leq   n^{1+1/2d}/2$, we have, 
\[
\inparen{e^d (2s/dt)^d}^t \leq (O(e/d))^{dt} \cdot n^t\, . 
\]
We also have 
\[
\left(\frac{n^2}{t}\right)^t \geq \left(n^{1 + 1/2d} \right)^t \, .
\]
For any constant $d$, these estimates contradict the inequality above, thus implying a lower bound of $\Omega(n^{1+1/2d})$ on $s$. 

The statement on the running time for constructing $A_n$ follows again from \autoref{lem:hard-over-C}.
\end{proof}

\subsection{Lower bounds for depth-$2$ linear circuits}
\label{subsec:depth-2-direct-sum}
The lower bounds of \autoref{thm:depth-d-complex} and \autoref{thm:depth-d-finite} apply to any constant depth. However, here we briefly remark that in the special case of $d = 2$ there is in fact a much simpler construction. As discussed in the introduction, for depth-$2$ linear circuits, the best lower bounds currently known is a lower bound of $\Omega\left(n\frac{\log^2n}{\log\log n}\right)$ based on the study of super-concentrator graphs in the work of Radhakrishnan and Ta-Shma~\cite{RTS00}. We now discuss  two constructions of matrices in quasi-polynomial time which improve upon this bound. More formally, we prove the following theorem. 
\begin{theorem}\label{thm:depth-2-quasipoly}
Let $c$ be any positive constant. Then, there is a family $\{A_n\}_{n \in \N}$ of $n \times n$ matrices with entries in $\N$ of bit complexity at most $\exp(O(\log^{2c + 1} n))$ such that $A_n$ can be constructed in time $\exp(O(\log^{2c + 1} n))$ and any depth-$2$ linear circuit over $\C$ computing $A_n$ has size at least $\Omega(n\log^c n)$. 
\end{theorem}
The first construction directly follows from~\autoref{lem:hard-over-C} when invoked with $t = 10\cdot  \log^{2c} n$. Once we have the matrices guaranteed by~\autoref{lem:hard-over-C}, we just follow  the proof of~\autoref{thm:depth-d-complex} as is by taking $d = 2$ and $t = 10\log^{2c} n$. We skip the technical details and now discuss the second construction, which is based on the following observation. 
\begin{observation}\label{obs:trivial hard matrix}
Let $\{A_n\}_{n \in \N}$ be a family of matrices where $(A_n)_{i, j} = 2^{2^{(n+1)(i-1) + j}}$. Then, any depth$-2$ linear circuit computing $A_n$ has size $\Omega(n^2)$. 
\end{observation}
\begin{proof}
The key to the proof is to observe that for $t = n^2/4$, $\Sigma_t(A_n) \geq 2^{\binom{n^2}{n^2/4}} \geq 2^{2^{n^2/2}}$.  This follows from the fact that each $t$ wise product of the entries of $A_n$ is a power of $2$ where the exponent is a sum of  powers of $2$ and  for any two distinct degree $t$ multilinear monomials in the entries of $A_n$, the set of powers of $2$ that appear in the exponent are distinct. On the other hand, from~\autoref{lem:ss-up-sigma}, we know that if $A_n$ can be computed by a depth-$2$ linear circuit of size at most $s$, then 
\[
\Sigma_t(A_n) \leq 2^{2n^3\left(e^2(4s/n^2) \right)^{n^2/4}} \, .
\]
Now, for $s \leq n^2/100$, this upper bound is much smaller than the lower bound of $2^{2^{n^2/2}}$. Thus, any depth-$2$ linear circuit for $A_n$ over $\C$ has size at least $n^2/100$. 
\end{proof}

If we directly use this observation to construct hard matrices,  the bit complexity of the entries of $A_n$ (and hence the time complexity of constructing $A_n$) is as large as $2^{\Theta(n^2)}$. However, it also gives a much stronger (quadratic) lower bound on the depth-$2$ linear circuit size for $A_n$ than what is promised in~\autoref{thm:depth-2-quasipoly}. For our second construction for hard matrices for~\autoref{thm:depth-2-quasipoly}, we invoke~\autoref{obs:trivial hard matrix} to construct \emph{small} hard matrices (thus saving on the running time) and then construct a larger block diagonal matrix by taking a Kronecker product of this small hard matrix with a large  identity matrix.  The following lemma then guarantees a non-trivial lower bound on the size of any depth-$2$ linear circuit computing this larger block diagonal matrix. 

\begin{lemma}\label{lem:block diagonal hard matrix}
Let $A$ be an $k \times k$ matrix, such that any depth-$2$ linear circuit computing $A$ has size at least $s$. Let $B$ be an $mk \times mk$ matrix defined as  $B = \mathbf{I}_m \otimes A $, where $\otimes$ denotes the Kronecker product, and $\mathbf{I}_m$ the $m \times m$ identity matrix. Then, any depth-$2$ linear circuit computing $B$ has size at least $m\cdot s$. 
\end{lemma}

\begin{proof}
A depth-$2$ linear circuit for $B$ gives a factorization of $B$ as $P\cdot Q$ for an $mk \times r$ matrix $P$ and an $r \times mk$ matrix $Q$ for some parameter $r$. We partition the rows of $P$ into $m$ contiguous blocks of size $k$ each, and let $P_i$ be the $k \times r$ submatrix which consists of the $i^{th}$ block (i.e. rows $(i-1)k + 1, (i-1)k + 2, \ldots, ik$ of $P$). Similarly, we partition  the columns of $Q$ into $m$ contiguous blocks of size $k$ each and let $Q_i$ be the $r \times k$ submatrix of $Q$ corresponding to the $i^{th}$ block. From the structure of $B$, it follows that  for every $i \in \{1, 2, \ldots, m\}$, $P_i\cdot Q_i = A$. From the lower bound on the size of any depth-$2$ linear circuit for $A$, we get that $\spars{P_i} + \spars{Q_i} \geq s$. Combining this lower bound for $i = 1, 2, \ldots, m$, we get $\spars{P} + \spars{Q} = \sum_{i = 1}^m\left(\spars{P_i} + \spars{Q_i}\right) \geq m\cdot s$.  
\end{proof}

We now note that~\autoref{obs:trivial hard matrix} and~\autoref{lem:block diagonal hard matrix} imply another family of matrices for which~\autoref{thm:depth-2-quasipoly} holds. 

\begin{proof}[Second proof of~\autoref{thm:depth-2-quasipoly}]
Pick $k = \Theta(\log^c n)$ such that $k$ divdes $n$, and let $M_k$ be the matrix defined as $(M_k)_{i, j} = 2^{2^{(k+1)(i-1) + j}}$. Let $A_n = \mathbf{I}_{n/k} \otimes M_k$.  Clearly, $A_n$ can be constructed in time $2^{O(k^2)}$. Moreover, from~\autoref{obs:trivial hard matrix} and~\autoref{lem:block diagonal hard matrix} it follows that any depth-$2$ linear circuit computing $A_n$ has size at least $\Omega(n/k \cdot k^2) = \Omega(n\log^ c n)$. 
\end{proof}

We note that even though the discussion in this section was confined to depth-$2$ linear circuit lower bounds over $\C$, similar ideas can be extended to other fields as well.

\subsubsection*{Extension of the direct sum based construction to arbitrary constant depth?}
In light of the above construction, it is a natural question is to ask if this idea also extends to the construction of hard matrices for depth-$d$ circuits for arbitrary constant $d$. While this is a reasonable conjecture, the easy proof of \autoref{lem:block diagonal hard matrix} breaks down even at depth $3$.

There are some variations of this idea, such us looking at $\mathbf{J}_{n/k} \otimes M_k$, where $\mathbf{J}$ is the all-1 matrix, which would work equally well to prove a lower bound for depth-$2$, but for which it is possible to prove an $O(n)$ upper bound in depth-$3$.

Furthermore, it can be seen that upper bounds on matrix multiplication in bounded depth will give small linear circuits for computing $\mathbf{I}_{n/k} \otimes M_k$. Thus, improved lower bounds using this construction, even for depth-$3$, will require proving new lower bounds for matrix multiplication in bounded depth (the current best lower bounds are again barely super-linear \cite{RS03}).

\section{Lower bounds via Hitting Sets}
\label{sec:lb-hitting-sets}

In this section, we prove lower bounds for several classes of depth 2 circuits using hitting sets for matrices. We first recall the definition.

\begin{definition}[Hitting set for matrices, \cite{FS12}]
Let $\calC \subseteq \F^{n \times n}$ be a set of matrices. A set $\calH \subseteq \F^n \times \F^n$ is said to be a \emph{hitting set} for $\calC$, if for every non-zero $C \in \calC$, there is a pair $(\veca, \vecb) \in \calH$ such that 
\[
\ip{\veca, M\cdot \vecb} = \sum_{i \in [n], j \in [m]} M_{i, j} a_i b_j \neq 0 . \qedhere
\]
\end{definition}

\subsection{Matrices with no sparse vectors in their kernel}\label{sec:hitting set const}
In this section, we recall some simple, deterministic and efficient constructions of matrices which do not have any sparse non-zero vector in their kernel. Such a construction forms the basic building block for building hard instances of matrices for various cases of the matrix factorization problem that we discuss in the rest of this paper. We start by describing such a construction over the field of real numbers. 

\subsubsection{Construction over $\R$}
The following is a weak form of a classical lemma of Descartes.

\begin{lemma}[Descartes' rule of signs]\label{lem:rule of signs}
Let $d_1< d_2 < \cdots < d_k$ be non-negative integers, and let $a_1, a_2, \ldots, a_k$ be arbitrary real numbers. Then, the number of distinct positive roots of the polynomial $\sum_{i = 1}^k a_i x^{d_i}$ is at most $k-1$.  
\end{lemma}
\autoref{lem:rule of signs} immediately gives the following construction of a small set of vectors, such that not all of them can lie in the kernel of any matrix with at least one sparse row. 
\begin{lemma}\label{lem:hitting set for real sparse}
For $i \in [n]$, let $\vecv_i := \inparen{1, i, i^2, \ldots, i^{n-1}} \in \R^n$. Then, for every $1 \le s \le n$ and for every $m \times n$ matrix $B$ over real numbers that has a non-zero row with at most $s$ non-zero entries, there is an $i \in [s]$ such that $B\cdot \vecv_i \neq \mathbf{0}$.  
\end{lemma} 
\begin{proof}
Let $(a_0, a_1, \ldots, a_{n-1}) \in \R^n$ be any non-zero vector with at most $s$ non zero entries. So, the polynomial $P(x) = \sum_{i = 0}^{n-1} a_i x^i$ has sparsity at most $s$. From~\autoref{lem:rule of signs}, it follows that $P$ has at most $t-1$ positive real roots. Therefore, there exists an $i \in [s]$ such that $i$ is \emph{not} a root of $P(x)$, i.e., $P(i) \neq 0$. The lemma now follows immediately  by taking $(a_0, a_1, \ldots, a_{n-1})$ to be any non-zero $s$-sparse row of $B$. 
\end{proof}
We remark that~\autoref{lem:hitting set for real sparse} also holds for matrices over $\C$ which have a sparse non-zero row for the choice of the vectors $v_i$ as above. This follows from the application of~\autoref{lem:rule of signs} separately for the real and complex parts of a sparse complex polynomial, both of which are individually sparse, with real coefficients and at least one of them is not identically zero. This observation extends  our results over $\R$ in~\autoref{sec:sym lb} to the field of complex numbers. 

\subsubsection{Construction over finite fields}
We now recall some basic properties of Reed-Solomon codes, and observe they can be used as well in lieu of the construction in \autoref{lem:hitting set for real sparse}. 

The proofs for these properties can be found in any standard reference on coding theory, e.g., Chapter 5 in~\cite{GRS-book}. 

\begin{definition}[Reed Solomon codes]\label{def:RS codes}
Let $\F_q =  \{\alpha_0, \alpha_2, \ldots, \alpha_{q-1}\}$ be the finite field with $q$ elements and let $k \in \{0, 1, \ldots, q-1\}$. The Reed-Solomon code of block length $q$ and dimension $k$ are defined as follows. 
\[
RS_{q}[q, k] = \{\left(P(\alpha_0), P(\alpha_1), \ldots, P(\alpha_{q-1})\right) : P(z) \in \F_q[z], \deg(P) \leq k-1 \}. \qedhere
\]
\end{definition}
\begin{lemma}\label{lem:RS props}
Let $\F_q$ be the finite field with $q$ elements and let $k \in \{0, 1, \ldots, q-1\}$. The linear space $RS_{q}[q, k]$ as in~\autoref{def:RS codes} satisfies the following properties.  
\begin{itemize}
\item Every non-zero vector in $RS_q[q, k]$ has at least $q-k +1$ non-zero coordinates. 
\item The dual of $RS_{q}[q, k]$ is the space of Reed Solomon codes of block length $q$ and dimension $q-k$. 
\end{itemize}
\end{lemma}

\begin{lemma}\label{lem:RS dual const}
Let $\F_q = \{\alpha_0, \alpha_2, \ldots, \alpha_{q-1}\}$ be the finite field with $q$ elements. For any $k\leq q-1$, let $G_k$ be the $q \times k$ matrix over $\F_q$ whose $i$-th row is $(1, \alpha_{i-1}, \alpha_{i-1}^2, \ldots, \alpha_{i -1}^{k-1})$. Then, every non-zero vector in $\F_q^q$ in the kernel of $(G_k)^{T}$ has at least $k + 1$ non-zero coordinates. 
\end{lemma}
 \begin{proof}
 Observe that $G_k$ is the precisely the generator matrix of Reed Solomon codes of block length $q$ and dimension $k$ over $\F_q$. In particular, the linear space $RS_q[q, k]$ as in~\autoref{lem:RS props} is spanned by the columns of $G_k$. Thus any vector $\vecw$ in the kernel of $(G_k)^{T}$ is in fact a codeword of  the dual of these codes, which as we know from Item 2 of~\autoref{lem:RS props}, is itself a Reed Solomon code of block length $q$ and dimension $q-k$. From the first item of~\autoref{lem:RS props}, it now follows that $\vecw$ has at least $k+1$ non-zero coordinates. 
 \end{proof}
The following lemma is an analog of~\autoref{lem:hitting set for real sparse}.  
 \begin{lemma}\label{lem:hitting set for finite fields sparse}
Let $\F_q = \{\alpha_0, \alpha_2, \ldots, \alpha_{q-1}\}$ be the finite field with $q$ elements, $s \in [q]$ be a parameter and let $\vecv_i$ be the $i$-th column of the matrix $G_{k}$ as in~\autoref{lem:RS dual const} for $k = s$.  

Then, for every $m \times n$ matrix $B$ over $\F_q$ that has a non-zero row with at most $s$ non zero entries, there is an $i \in [s]$ such that $B\cdot \vecv_i \neq 0$.  
\end{lemma} 
\begin{proof}
The proof follows from the observation that any non-zero vector orthogonal to all the vectors $v_1, v_2, \ldots, v_s$ must be in the kernel of the matrix $G_{s}^T$ and hence by~\autoref{lem:RS dual const} must have at least $s+1$ non-zero entries. 
\end{proof}

\subsection{Lower bounds for symmetric circuits}\label{sec:sym lb}

We now prove our lower bounds for symmetric circuits. Recall that a symmetric circuit is a linear depth-2 circuit of the form $B^TB$.

\begin{theorem}\label{thm:sym factor lb}
There is an explicit family of positive semidefinite matrices $\{M_n\}$ such that every symmetric circuit computing $M_n$ has size at least $n^2/4$.
\end{theorem}

For the proof of this theorem, we give an efficient deterministic construction of a hitting set $\calH$ for the set of matrices which factor as $B^T\cdot B$ for $B$ of sparsity less than $n^2/4$, and
as outlined in \autoref{sec:techniques}, we construct a hard matrix $M = \tilde{M}^T\cdot \tilde{M}$ which is not hit by such a hitting set and has a high rank.

We start by describing the construction of $M$.

\begin{lemma}
\label{lem:hard-PSD-matrix}
Let  $\set{\vecv_i : i \in [n]}$ be the set of vectors defined in \autoref{lem:hitting set for real sparse}. 
There exists an explicit PSD matrix $M$ of rank $n/2$ such that $\vecv_i^T M\vecv_i = 0$ for $i \in [n/2]$.
\end{lemma}

\begin{proof}
We wish to find a matrix $\tilde{M}$ of high rank such that $\tilde{M} \vecv_i = 0$ for $i=1,\ldots,n/2$. This can be done by completing $\{\vecv_i : i \in \{1, 2, \ldots, n/2\}\}$ to a basis (in an arbitrary way) and requiring that the other $n/2$ basis elements are mapped to linearly independent vectors under $\tilde{M}$. Conveniently, the set $\set{\vecv_i : i \in [n]}$ is itself a basis for $\R^n$: the matrix $V$ whose rows are the $\vecv_i$'s is a Vandermonde matrix.

We now describe this in some more detail. For $i \in [n]$, let $\vece_i$ by the $i$-th elementary basis vector.
For a set of $n^2$ variables $Y = (y_{i, j})_{n \times n}$ consider the system of (non-homogeneous) linear equations on the variables $Y$ given by the $n$ constraints.
\begin{align*}
 Y\cdot \vecv_i &= 0 \quad\; \text{for } i \in \{1, 2, \ldots,  n/2\}   \\
  Y\cdot \vecv_i &= \vece_i \quad \text{for } i \in \{ n/2  + 1, \ldots, n\} \, .
\end{align*}

Since the vectors $\set{\vecv_i : i \in [n]}$ are linearly independent, this system has a solution, which can be found in polynomial time using basic linear algebra. More explicitly the $j$-th row of $Y$, $\vecy_j$, is given by the solution to the linear system $V\cdot (\vecy_j)^T = 0$ for $1 \le j \le n/2$ and $V \cdot (\vecy_j)^T = \vece_j$ for $n/2+1 \le j \le n$ where $V$ is the Vandermonde matrix whose rows are the $\vecv_i$'s. Let $\tilde{M}$ be the matrix whose rows are the solution to the system above. Also, note that the rank of $\tilde{M}$ is at least $n/2$, as linearly independent vectors $\vece_{n/2+1}, \vece_{n/2+2}, \ldots, \vece_{n}$ are in the image of the linear transformation given by $\tilde{M}$. 

Now let $M = (\tilde{M}^T) \cdot \tilde{M}$, so that indeed $M$ is a positive semi-definite matrix, and $\rank M = n/2$ as well. It immediately follows that
\[
\vecv_i^T M \vecv_i = (\vecv_i^T \tilde{M}^T)(\tilde{M}\vecv_i) = 0. \qedhere
\]
\end{proof}

We are now ready to prove \autoref{thm:sym factor lb}.

\begin{proof}[Proof of~\autoref{thm:sym factor lb}]
Let $M$ be the matrix from \autoref{lem:hard-PSD-matrix}. Let $B \in \R^{m \times n}$ be real matrix such that $\spars{B} < n^2/4$, and suppose towards contradiction that $M = B^TB$.

It follows that the rank of $B$ must be at least $n/2$. Thus, $B$ must have at least $n/2$ non-zero rows. Now, since the total sparsity of $B$ is at most $n^2/4-1$, there must be a non-zero row of $B$ with sparsity at most $(n^2/4-1)/(n/2) \leq  n/2$.  From~\autoref{lem:hitting set for real sparse}, it follows that there is an $i \in [n/2]$ such that $B\cdot \vecv_i$ is non-zero. Thus, for this index $i$, we have that
\[
\vecv_i^T (B^TB)\vecv_i = \left\Vert B\vecv_i \right\Vert_2^2 \neq 0,
\]
contradicting \autoref{lem:hard-PSD-matrix}.
\end{proof}

We remark that the proof of \autoref{thm:sym factor lb} goes through almost verbatim for symmetric circuits over $\C$ (recall that over $\C$ these are circuits of form $B^*B$, where $B^*$ is the conjugate transpose of $B$). 

\subsection{Lower bounds for invertible circuits}

Recall that an invertible circuit is a circuit of them form $BC$ where either $B$ or $C$ is invertible. In this section, we prove~\autoref{thm:intro:invertible}, which shows a quadratic lower bound for such circuits. For convenience, we restate the theorem.

\begin{theorem}
\label{thm:invertible-lb}
There exists an explicit family of $n \times n$ matrices $\set{A_n}$, over any field $\F$ such that $\F \ge \poly(n)$, such that every invertible circuit computing $A_n$ has size $n^2/4$.
\end{theorem}

\begin{proof}
We give a proof over the field of real numbers  and highlight the ideas necessary to extend the argument to work over large enough finite fields. 

Fix $n$, and let $M = \tilde{M}^T \tilde{M}$ be the matrix constructed in \autoref{lem:hard-PSD-matrix}. Let $B$ and $C$  be $n \times n$ matrices over $\R$ such that $M=BC$. Suppose first that $B$ is invertible and $C$ has sparsity less than $n^2/4$.

Since $\rank(M) \ge n/2$, the same applies for $\rank(C)$, and hence the number of non-zero rows in $C$ must be at least $n/2$. Thus, $C$ must have a non-zero row with at most $(n^2/4-1)/(n/2) \leq n/2$ non-zero entries.  Along with \autoref{lem:hitting set for real sparse}, this implies that there is an $i \in [n/2]$ such that $C\cdot \vecv_i \neq \mathbf{0}$, where $\vecv_i$ is as in \autoref{lem:hitting set for real sparse}. Since $B$ is invertible, we get that $(B\cdot C \cdot \vecv_i)$ is a non-zero vector, so for some $j \in [n]$,
\[
\vece_j^T (BC)\vecv_i \neq 0.
\]
However, as in the proof of \autoref{lem:hard-PSD-matrix}
\[
\vece_j^T (M)\vecv_i = \vece_j^T \tilde{M}^T \tilde{M} \vecv_i = 0,
\]
since $\tilde{M} \vecv_i = 0$ for all $i \in [n/2]$.

The case that $B$ is sparse and $C$ is invertible is virtually the same, by considering $\vecv_i^T (BC) \vece_j$, and replacing the argument on the rows of $C$ by a similar one on the columns of $B$.

For the proof over finite fields, we replace every application of \autoref{lem:hitting set for real sparse} by \autoref{lem:hitting set for finite fields sparse}. Note that this requires the $n$-th matrix in the family to be defined over a field of size more than $n$. The rest of the argument essentially remains the same. 
\end{proof}

Over fixed finite fields (for example, $\F_2$), it is possible to prove an analog of \autoref{thm:invertible-lb}, with worse constants, by replacing the use of Reed-Solomon codes with any good explicit error-correcting code $C$ of dimension $\alpha n$ and distance $\delta n$ for some fixed constants $\alpha, \delta >0$. The proof proceeds as above by finding a matrix $\tilde{M}$ of rank $\alpha n$ such that $M\vecv= 0$ for every $\vecv \in C^{\perp}$.

\section{Open Problems}

An important problem that continues to remain  open  is to prove a lower bound of the form $\Omega(n^{1+\varepsilon})$ for some constant $\varepsilon>0$ for the depth-2 complexity of an explicit matrix. Such a lower bound would follow from an explicit hitting set of size at most $n^2 - 1$ for the class of polynomials of the form $\vecx^T BC \vecy$ such that $\spars{B}+\spars{C} \le n^{1+\varepsilon}$.

Another natural question here is be to understand if this PIT based approach can be used for explicit  constructions of rigid matrices, which improve the state of art. One concrete question in this direction would be to construct explicit hitting sets for the set of matrices which are \emph{not} $(r, s)$ rigid for $rs > \omega(n^2\log (n/r))$. Using the techniques in this paper, it is possible to construct hitting sets of size $O(rs)$ for matrices which are not $(r, s)$ rigid. But, this is non-trivial only when $rs \le cn^2$ for some constant $c<1$, which is a regime of parameters for which explicit construction of rigid matrices is already known. A sequence of recent results~\cite{AW17, DE17, DL19} showed that many natural candidates for rigid matrices that posses certain symmetries are in fact not as rigid as suspected. This approach might circumvent these obstacles by giving an explicit construction which is not ruled out by these results.

A lower bound of $s$ on the size of depth $d$ linear circuits computing the linear transformation $A\vecx$ implies a lower bound of $\Omega(s)$ for depth $\Omega(d)$ algebraic circuits computing the  degree-2 polynomial $\vecy^T A \vecx$ \cite{BS83, KalSin91} (so, we can convert lower bounds for circuits with $n$ outputs to lower bounds for circuits with 1 output). A notable open problem in algebraic complexity, which is very related to this work, is to prove any super-linear lower bound for algebraic circuits of depth $O(\log n)$ computing a polynomial with constant total degree. We refer to \cite{Raz10a} for a discussion on the importance of this problem.

\section*{Acknowledgements}
We thank Swastik Kopparty for an insightful discussion on  explicit construction of Sidon sets over finite fields. We also thank Rohit Gurjar, Nutan Limaye, Srikanth Srinivasan and Joel Tropp for helpful discussions. 

\bibliographystyle{customurlbst/alphaurlpp}
\bibliography{references}


\end{document}